\def\be{\begin{equation}}
    \def\ee{\end{equation}}
\def\ba{\begin{array}}
    \def\ea{\end{array}}

\documentclass[prl,showpacs,twocolumn,amsmath]{revtex4}
\usepackage{amsmath}
\usepackage{amsfonts}
\usepackage{mathrsfs}
\usepackage{amssymb}
\usepackage{pifont}
\usepackage{epsfig,subfigure,dsfont,amsthm,amsbsy,mathrsfs,amscd}
\usepackage{epstopdf}
\usepackage{color}
\def\qed{\leavevmode\unskip\penalty9999 \hbox{}\nobreak\hfill
    \quad\hbox{\leavevmode  \hbox to.77778em{%
            \hfil\vrule   \vbox to.675em%
            {\hrule width.6em\vfil\hrule}\vrule\hfil}}
    \par\vskip3pt}
\usepackage{leftidx}
\newtheorem{theorem}{Theorem}

\newtheorem{lemma}{Lemma}
\begin{document}
\title{\large\bf Moments based entanglement criteria and measures}

\author{Yiding Wang$^{1}$, Tinggui Zhang$^{1, \dag}$, Xiaofen Huang$^{1}$ and Shao-Ming Fei$^{2}$}
    \affiliation{ ${1}$ School of Mathematics and Statistics, Hainan Normal University, Haikou, 571158, China \\
        $2$ School of Mathematical Sciences, Capital Normal University, Beijing 100048, China \\
        $^{\dag}$ Correspondence to tinggui333@163.com}

\bigskip
\bigskip

\begin{abstract}
Quantum entanglement plays a key role in quantum computation and
quantum information processing. It is of great significance to find
efficient and experimentally friend separability criteria to detect
entanglement. In this paper, we firstly propose two easily used
entanglement criteria based on matrix moments. The first
entanglement criterion only uses the first two realignment moments
of a density matrix. The second entanglement criterion is based on
the moments related to the partially transposed matrix. By detailed
examples we illustrate the effectiveness of these criteria in
detecting entanglement. Moreover, we provide an experimentally measurable lower bound of
concurrence based on these moments. Finally, we present both bipartite and genuine
tripartite entanglement measures based on the moments of the reduced states. By detailed examples, we show that our entanglement measures characterize the quantum entanglement in a more fine ways than the existing measures.
\end{abstract}

 \maketitle

\section{I. Introduction}
Quantum entanglement \cite{ab} is a novel characteristic of quantum mechanics and plays an important role in many quantum tasks such as
quantum communications \cite{cgcraw,rh,nr,j}, quantum simulation \cite{s}, quantum computing \cite{asc,ar,mc} and quantum cryptography \cite{a,ngwh,l}.
In this context, detecting the quantum entanglement has become particularly important.

Let $H_A$ and $H_B$ denote the Hilbert spaces of systems $A$ and $B$ with
dimensions $m$ and $n$, respectively. A quantum state $\rho\in H_A\otimes H_B$ is separable if it can be expressed as a convex combination of product states,
\begin{center}
$\rho=\sum_{i}p_i\rho_i^A \otimes \rho_i^B,\ \ \sum_{i}p_i=1,\ \ 0\leq p_i\leq1.$
\end{center}
Otherwise, the state $\rho$ is entangled. Generally it is a
challenge to detect the entanglement for a given state. The PPT
criterion \cite{ap} is both necessary and sufficient for the
separability of quantum states in systems $2\otimes 2$ and $2\otimes
3$ \cite{abn}. This criterion indicates that for any bipartite
separable state $\rho$, the matrix $\rho^{\tau}$ obtained from
partial transpose with respect to subsystem $B$ is still positive
semi-definite, where $(\rho^{\tau})_{ij,kl}=(\rho)_{il,kj}$. Any
state that violates the PPT criterion is an entangled one. The
realignment is another permutation of the elements of a density
matrix. The realignment criterion \cite{cw,oru} says that for any
bipartite separable state $\rho$, the trace norm of the realigned
matrix $\rho^R$ is not greater than 1, i.e., $\|\rho^R\|\leq1$,
where $(\rho^R)_{ij,kl}=(\rho)_{ik,jl}$, and the trace norm of an
operator $E$ is defined by $\|E\|:=Tr(\sqrt{E^{\dagger}E})$. A state
is entangled if it violates the realignment criterion.

There are also many other approaches to detect the entanglement. The
entanglement witnesses can be used to detect the entanglement
\cite{b,rphk,oggt} experimentally, although the construction of the
witness operators generally requires the prior deterministic
information of the quantum state. Locally randomized
measurements \cite{snao,ano,abcp,tap,ljw} and parameterized
entanglement monotone \cite{xmys,zwms,zwsm,jfjq,htfy} have been also
adopted to detect entanglement. Besides, the quantum entanglement is also studied based on the truncated moment problem that is well studied mathematically. Bohnet et al. proposed a necessary and sufficient condition of separability that can be applied by using a hierarchy of semi-definite programs \cite{bbg}.

Recently, the authors in \cite{arhr,yso} show that the first three
partially transpose (PT) moments can be used to detect entanglement.
The advantage of the PT moments is that they can be experimentally
measured through global random unitary matrices \cite{zzl,jlas} or
local randomized measurements \cite{arhr} based on quantum shadow
estimation \cite{hrj}. In \cite{yso} the authors proposed a
separability criterion based on PT moments called $p_3$-OPPT
criterion. Neven et al. proposed an ordered set of experimentally measured conditions for detecting entanglement \cite{ncvk}, with the $k$-th condition given by comparing the moments of the PT density operator up to order $k$. Zhang et al. introduced a separability criterion based on the rearrangement moments \cite{zjf}. In \cite{kzw} the authors
introduced $\Lambda$-moments with respect to any positive maps
$\Lambda$. They showed that these $\Lambda$-moments can effectively
characterize the entanglement of unknown quantum states without
prior reconstructions. In \cite{zyhp}, the authors proposed a
framework for designing multipartite entanglement criteria based on
permutation moments. The author in \cite{ali} demonstrates that for two-qubit quantum systems the PT moments can be expressed as functions of principal minors and shows that the PT moments can detect all the negative partial transpose entanglement of GHZ and W states mixed with white noise. A separability criterion and its physical realization has been also proposed
by using the moments of the realigned density matrices \cite{saak,ssa}.

Besides the separability, the quantification of entanglement is also of great significance \cite{ngmm}. Some entanglement measures have been
proposed to quantify the entanglement \cite{rphk,oggt,wkw,vw,sk,fma},
among which one of the most well known measures is the concurrence
\cite{rphk,oggt,wkw}. Let $|\psi_{AB}\rangle$ be a bipartite pure
state in $H_A\otimes H_B$. The concurrence of $|\psi_{AB}\rangle$ is given by
\begin{equation}
C(|\psi_{AB}\rangle)=\sqrt{2[1-Tr(\rho^2_A)]},
\end{equation}
where $\rho_A=Tr_B(|\psi_{AB}\rangle\langle\psi_{AB}|)$ is the reduced density matrix. The concurrence for general bipartite mixed states $\rho$ is given by the convex-roof extension,
\begin{equation}
C(\rho)=\mathop{\min}_{\{p_i,|\psi_i\rangle\}}\sum_{i}p_iC(|\psi_i\rangle),
\end{equation}
where $p_i\geq 0$, $\sum_{i}p_i=1$ and the minimum is taking over all possible pure state decompositions of $\rho=\sum_{i}p_i|\psi_i\rangle_{AB}\langle\psi_i|$.

For multipartite systems, the quantification of the genuine multipartite entanglement remains a challenging problem. The authors in \cite{mcc}
proposed a genuine multipartite entanglement measure (GMEM)
based on the concurrences under bi-partitions. The authors in \cite{xe} introduced a genuine
three-qubit entanglement measure in terms of the
area of a triangle with the three edges given by bipartite
concurrences. More genuine multipartite
entanglement measures have been also presented \cite{eb,srpr,ss,hhe}. In \cite{srpr} the
authors proposed the generalized geometric measure. Further genuine multipartite
concurrences are studied in \cite{hhe}. Guo et al. \cite{gy} gave an approach of constituting
genuine $m$-partite entanglement measures from any bipartite
entanglement and any $k$-partite entanglement measure for $3\leq
k<m$. Recently, the authors in \cite{jzyy} constructed a proper
genuine multipartite entanglement measure by using the geometric
mean area of concurrence triangles according
to a series of inequalities related to entanglement distribution.

In this paper, we first propose two separability criteria based on moments, and illustrate their effectiveness in entanglement detection by specific examples. We then provide an experimentally measurable lower bound of concurrence based on the moments. We present a bipartite entanglement measure based on the moments of the reduced states. Furthermore, we propose a genuine tripartite entanglement measure based on our bipartite entanglement measure.
The paper is organized as follows. In the second section, we provide a separability criterion based on realignment moments. In the third section, we propose a separability criterion based on PT moments. In the fourth section, we derive an experimentally measurable lower bound of concurrence for arbitrary bipartite states. In the fifth section, we propose a bipartite entanglement measure based on reduced moments. In the sixth section, we put forward a genuine tripartite entanglement measure based on our bipartite entanglement measure.
We summarize and discuss our conclusions in the last section.

\section{II. Separability criterion based on realignment moments}

We first recall the realignment moments of density matrices. Let $\rho$ be a bipartite state in $H_A\otimes H_B$. The realignment moments are given by
\begin{center}
$T^R_k=Tr[(\rho^{R\dagger}\rho^R)^k],~~~k=1,2,...,mn.$
\end{center}
Let $\sigma_1,\sigma_2,...\sigma_d$ be the $d$ nonzero singular values of $\rho^R$. We have
\begin{align}
T^R_1 &= Tr(\rho^{R\dagger}\rho^R)=\sum_{i}^{d}\sigma^2_i,\\
T^R_2 &= Tr[(\rho^{R\dagger}\rho^R)^2]=\sum_{i}^{d}\sigma^4_i.
\end{align}
We have the following conclusion on the separability of $\rho$ in terms of the realignment moments.

\begin{theorem}
If a state $\rho$ is separable, then $Q\leq 1$, where $Q\equiv \sqrt{\sqrt{2[(T^R_1)^2-T^R_2]}+T^R_1}$.
\end{theorem}

\begin{proof} By the definition we have
		\begin{eqnarray}\label{key1}
			T^R_2 &=& (T^R_1)^2-2\sum_{i<j}\sigma^2_i\sigma^2_j\nonumber\\&\geq& (T^R_1)^2-2(\sum_{i<j}\sigma_i\sigma_j)^2\nonumber\\&=& (T^R_1)^2-\frac{1}{2}(2\sum_{i<j}\sigma_i\sigma_j)^2\nonumber\\&=& (T^R_1)^2-\frac{1}{2}[(\sum_{i=1}^{d}\sigma_i)^2-\sum_{i=1}^{d}\sigma_i^2]^2\nonumber\\&=& (T^R_1)^2-\frac{1}{2}(\|\rho^R\|^2-T^R_1)^2,
	\end{eqnarray}
where the inequality is due to the following fact: for non negative
real numbers $x_1,x_2,...,x_n$, $\sum_{i=1}^{n}x^2_i\leq
(\sum_{i=1}^{n}x_i)^2$. The relation Eq. (\ref{key1}) implies that
    \begin{align*}
        & (\|\rho^R\|^2-T^R_1)^2 \geq 2[(T^R_1)^2-T^R_2]\nonumber
        \\ & \Leftrightarrow \|\rho^R\|^2-T^R_1 \geq \sqrt{2[(T^R_1)^2-T^R_2]}.
        \end{align*}
Therefore, we have
$$\|\rho^R\|\geq \sqrt{\sqrt{2[(T^R_1)^2-T^R_2]}+T^R_1}.$$
According to the realignment criterion, if a quantum state $\rho$ is separable, $\|\rho^R\|\leq 1$, which completes the proof.
\end{proof}

From Theorem 1 a quantum state which violates the inequality $Q\leq 1$ must be entangled.
The advantage of our criterion is its simplicity as it only involves the first two moments of the realigned matrix. To verify the efficiency of our criterion let us consider the following example given in \cite{ga}.

Example 1.
$$
\rho_a=\left(\begin{array}{ccccccccc}
    \frac{1-a}{2} &  0 & 0 & 0 &  0 & 0 & 0 & 0 & -\frac{11}{50} \\
    0 & 0 & 0 & 0 &  0 & 0 & 0 & 0 & 0 \\
    0 & 0 & 0 & 0 &  0 & 0 & 0 & 0 & 0\\
    0 & 0 & 0 & 0 &  0 & 0 & 0 & 0 & 0 \\
    0 & 0 & 0 & 0 & \frac{1}{2}-a & -\frac{11}{50} & 0 & 0 & 0 \\
    0 & 0 & 0 & 0 & -\frac{11}{50} & a & 0 & 0 & 0 \\
    0 & 0 & 0 & 0 & 0 & 0 & 0 & 0 & 0 \\
    0 & 0 & 0 & 0 & 0 & 0 & 0 & 0 & 0 \\
    -\frac{11}{50} & 0 & 0 & 0 & 0 & 0 & 0 & 0 & \frac{a}{2}
\end{array}\right),
$$
where $\frac{1}{50}(25-\sqrt{141}) \leq a \leq
\frac{1}{100}(25+\sqrt{141})$. The first two
realignment moments of $\rho_a$ are
\begin{align*}
T^R_1 &= \frac{7a^2}{4}-a+\frac{867}{1250}, \\
T^R_2 &= \frac{35a^4}{16}-\frac{3a^3}{2}+\frac{373a^2}{250}-\frac{373a}{625}+\frac{292899}{1562500}.
\end{align*}
We obtain that when $\frac{1}{50}(25-\sqrt{141}) \leq a \leq
\frac{1}{100}(25+\sqrt{141})$, the inequality in Theorem 1 is
violated. That is, our criterion can detect all the entanglement for
this family of states. See Figure. 1.
\begin{figure}[htbp]
	\centering
	\includegraphics[width=0.5\textwidth]{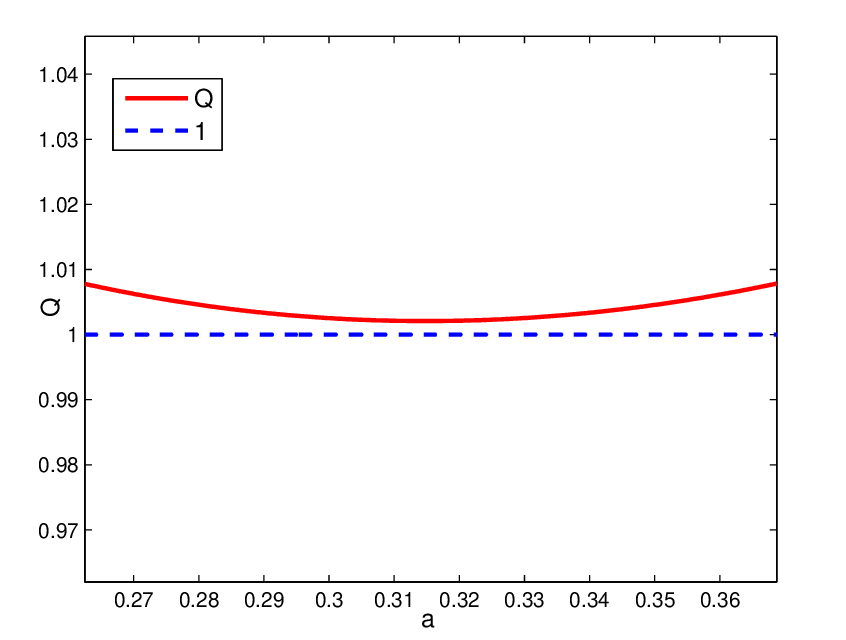}
	\vspace{-2.0em} \caption{The red solid line represents the value
		of $Q$. When $\frac{1}{50}(25-\sqrt{141}) \leq a \leq
		\frac{1}{100}(25+\sqrt{141})$, there is always $Q>1$, which means that this family of quantum states violates Theorem 1.} \label{Fig.1}
\end{figure}

In the above example, our entanglement criterion and realignment criterion are equally effective, as they both detect all entangled states in this family of quantum states. However, this is not always the case. In general, our criterion is weaker than the realignment criterion because our criterion is derived from the latter.

Example 2. Let us consider the Werner state,
$$
\rho_u=u|\psi\rangle\langle\psi|+\frac{1-u}{4}I_4,
$$
where $0\leq u\leq1$, $|\psi\rangle=\frac{1}{\sqrt{2}}(|00\rangle+|11\rangle)$ and $I_4$ is the $4\times4$ identity matrix. By calculation it can be concluded that $Q>1$ when $u>0.54$, which means that entanglement of $\rho_u$ can be detected by our entanglement criterion within the range of $0.54<u\leq1$. However, according to the realignment criterion $\rho_u$ is entangled when $u>\frac{1}{3}$. This also indicates that in order to achieve the experimental feasibility, our criterion is weaker than the original realignment criterion.

\section{III. Separability criterion based on PT moments}
With respect to the partially transposed matrix $\rho^{\tau}$ of $\rho$, the PT moments are defined as
\begin{center}
$T^\tau_k=Tr[(\rho^\tau)^k],~~~k=1,2,...,mn.$
\end{center}
Consider the characteristic polynomial of $\rho^\tau$,
\begin{center}
$a_0\lambda^p-a_1\lambda^{p-1}+...+(-1)^{p-1}a_{p-1}\lambda+(-1)^{p}a_p$,
\end{center}
where $p=mn$ is the number of rows of the matrix $\rho^\tau$,
$a_0=1$, $a_k=\sum_{\{s_k\}\in S}\prod_{j\in s_k}\lambda_j$,
$k=1,2,...,p,$ ${s_k}$ denotes a subset of $S=\{1,2,...,p\}$
with $k$ elements. The characteristic polynomial coefficients and
the PT moments have the following relations \cite{jfjq},
\begin{equation}\label{ak}
a_{k+1}=\frac1{k+1}\sum_{l=0}^{k}(-1)^la_{k-l}T^\tau_{l+1}
\end{equation}
for $k=0,1,...,p-1$. We have the following result.

\begin{theorem}
If the bipartite state $\rho$ is separable, then
\begin{equation}\label{t2}
a_k\,a_{k+1}>0,~~~k=0,1,...,q-1,
\end{equation}
where $q$ is the rank of the matrix $\rho^\tau$, $a_k$ is given in
Eq.(\ref{ak}), with $a_q\neq 0$ and $a_{r}=0$ for $r>q$.
\end{theorem}

\begin{proof}
The characteristic polynomial of $\rho^{\tau}$ can be rewritten as
	$P(\lambda)=a_0\lambda^p-a_1\lambda^{p-1}+...+(-1)^{q}a_{q}\lambda^{p-q}.$
	We first prove that $\rho^{\tau}$ is positive semidefinite if and
	only if $a_ka_{k+1}>0$ for each $k=0,...,q-1$. If $a_ka_{k+1}>0$ for
	each $k=0,...,q-1$, since $a_0=1$ the symbols of the coefficients of
	the characteristic polynomial are strictly alternating. Thus
	$P(\lambda)$ has no negative roots. Otherwise, if we assume the existence of negative roots, we obtain contradictions. Hence $\rho^{\tau}$ has only
	nonnegative eigenvalues.

Conversely, if $\rho^{\tau}$ is positive semidefinite, we denote its
positive eigenvalues by $\lambda_1,...,\lambda_q$, with all the
remaining $p-q$ eigenvalues being $0$. Through inductive argument,
we obtain that the signs of the coefficients of the polynomials
$(\lambda-\lambda_1)(\lambda-\lambda_2)...(\lambda-\lambda_q)$
alternate strictly, which gives $P(\lambda)$ up to a factor
$\lambda^{p-q}$. Therefore, $\rho^{\tau}$ is positive semidefinite
if and only if $a_ka_{k+1}>0$ for each $k=0,...,q-1$. From the PPT
criterion that $\rho^{\tau}$ is positive semidefinite for any
bipartite separable state $\rho$, we complete the proof of Theorem
2.
\end{proof}

Theorem 2 implies that if a bipartite quantum state violates any
inequality in Eq.(\ref{t2}), it must be entangled. From the proof of
Theorem 2, it is seen that our criterion is equivalent to the PPT
criterion. However, the PPT criterion can not be applied without
state tomography. Our criterion can be used to detect the
entanglement of unknown quantum states. We only need to measure the
PT moments, since the conditions $a_ka_{k+1}>0$, $k=0,1,...,m-1$, in
the Theorem 2 can be transformed into the relationship among the
moments. We illustrate the usefulness of our criterion through the
following example.

Example 3. Consider the two-qubit isotropic state given in \cite{zlfw},
$$
\rho_b=\frac{1-b}{3}I_2\otimes I_2+\frac{4b-1}{3}|\psi\rangle\langle\psi|,~~0\leq b\leq 1,
$$
where $I_2$ denotes the second-order identity matrix, $|\psi\rangle=\frac{1}{\sqrt{2}}(|00\rangle +|11\rangle)$. We have
\begin{align*}
T^{\tau}_1 &= 1,\\T^{\tau}_2 &= \frac{1}{3}(4b^2-2b+1),\\T^{\tau}_3
&=
-\frac{8}{9}b^3+\frac{5}{3}b^2-\frac{2}{3}b+\frac{5}{36},\\T^{\tau}_4
&=\frac{84}{81}b^4-\frac{156}{81}b^3+\frac{126}{81}b^2-\frac{39}{81}b+\frac{21}{324}.
\end{align*}
Substituting the above moments into the inequalities in Theorem 2,
we obtain that $\rho_b$ is entangled when $b>0.5$, which is exactly
the same result as the one from the realignment and PPT criterion
directly, and stronger than the result $b \geq 0.608594$ given in
\cite{ssa}.

\section{IV. Experimentally measurable lower bound of concurrence}

For any $m\otimes n$ $(m\leq n)$ quantum state $\rho$, Chen et al.
proposed a lower bound of concurrence \cite{csf},
\begin{equation}\label{lc}
C(\rho)\geq \sqrt{\frac{2}{m(m-1)}}\max(\|\rho^{\tau}\|-1,\|\rho^R\|-1).
\end{equation}
To obtain experimentally measurable lower bound of concurrence,
we next derive the lower bounds according to the moments from $\|\rho^{\tau}\|$ and $\|\rho^R\|$.

\begin{theorem}
For any $m\otimes n(m\leq n)$ quantum state $\rho$, we have the following experimentally measurable lower bound of concurrence,
\begin{equation}\label{2c}
C(\rho)\geq \sqrt{\frac{2}{m(m-1)}} \max\{M_1,M_2,0\},
\end{equation} where
\begin{align*}
    & M_1=\sqrt{\sqrt{2[(T_1)^2-T_2]}+T_1}-1,\\
    & M_2=\sqrt{\sqrt{2[(T^R_1)^2-T^R_2]}+T^R_1}-1
\end{align*}
with $T_i=Tr[(\rho^{\tau \dagger}\rho^{\tau })^i]$ and $T^R_i=Tr[(\rho^{R\dagger}\rho^R)^i]$, $i=1,2$.
\end{theorem}

\begin{proof}
Firstly, we have proven that if $\rho$ is a separable state, then
$M_1\leq 0$. Similar to the proof of Theorem 1, we have
$\|\rho^{\tau}\|\geq \sqrt{\sqrt{2[(T_1)^2-T_2]}+T_1}$. Hence we
only need to prove that $\|\rho^{\tau}\| \leq 1$. Since $\rho$ is
separable, the eigenvalues $\xi_i$ of $\rho^{\tau}$ are non negative
and the summation of the eigenvalues is $1$, $\sum_{i} \xi_i=1$,
$\xi_i\geq 0$, $i=1,2,...,mn$. Hence the eigenvalues of
$\rho^{\tau\dagger}\rho^{\tau}$ are $\xi^2_i$ $(i=1,2,...,mn)$. As
the singular values of $\rho^{\tau}$ are the arithmetic square root
of the non negative eigenvalues of $\rho^{\tau\dagger}\rho^{\tau}$,
we have $\|\rho^{\tau}\|=\sum_{i} \xi_i=1$. From the definition of
concurrence and the formula (\ref{lc}), we obtain Eq.(\ref{2c}).
\end{proof}

Example 4. Consider the following $3\times 3$ dimensional quantum states
\begin{equation}
\rho_s=\frac{1-s}{9}I_9+s|\psi_3\rangle\langle\psi_3|,~~ s\in[0,1],
\nonumber
\end{equation}
where $|\psi_3\rangle=\frac{1}{\sqrt{3}}\sum_{i=1}^{3}|ii\rangle$.
The state is shown to entangled for $s>\frac{1}{4}$ \cite{mhph}.
From Theorem 3 we obtain the experimental measurable lower bound,
which detects most of the entangled states in this family, see
Figure. 2.
\begin{figure}[htbp]
\centering
\includegraphics[width=0.5\textwidth,height=0.36\textheight]{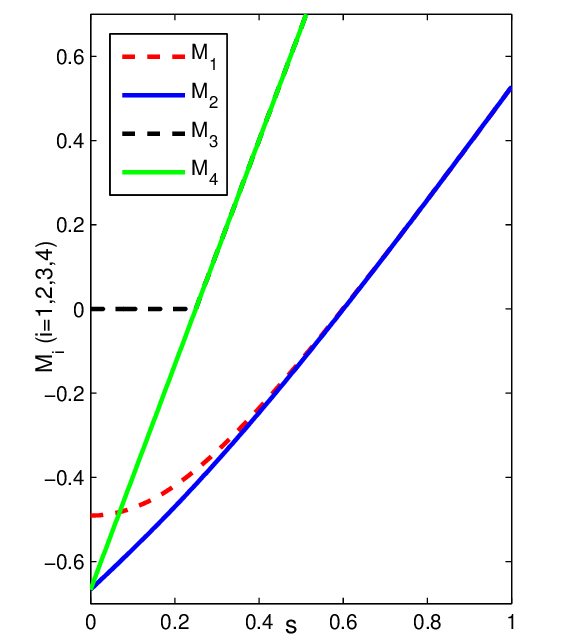}
\vspace{-1.7em} \caption{The red dashed line represents the value
		of $M_1$, the blue solid line denotes the value of $M_2$, the value of $\|\rho^{\tau}\|-1$ is represented by a black dashed line $M_3$, and the value of $\|\rho^R\|-1$ is represented by a green solid line $M_4$. For
		$s>0.5994$ the maximum values of $M_1$ and $M_2$ are greater than
		0.} \label{Fig.2}
\end{figure}

\section{V. Entanglement measure based on moments of reduced states}

Let $|\psi_{AB}\rangle=\sum_{i=1}^{d}\sqrt{\mu_i}|ii\rangle$ be a bipartite pure state in $H_A\otimes H_B$ in Schmidt decomposition, where $\sum_{i=1}^{d}\mu_i=1$, $\mu_i\geq 0$ $(i=1,2,...,d)$
and $d=\min(m,n)$ with $m$ and $n$ the dimensions of $H_A$ and $H_B$, respectively. Consider the characteristic polynomial of the reduced density matrix $\rho_A$
of $|\psi_{AB}\rangle$,
\begin{equation}
b_0\mu^m-b_1\mu^{m-1}+...+(-1)^{m-1}b_{m-1}\mu+(-1)^mb_m,
\end{equation}
where $b_0=1$,
\begin{equation}
b_k=\sum_{\{g_k\}\in G}\prod_{j\in g_k}\mu_j,~~k=1,2,...,m,
\end{equation}
with ${g_k}$ a subset of $G=\{1,2,...,m\}$ of $k$ elements.

The coefficients of the characteristic polynomial of a reduced density matrix for a bipartite pure state can be linearly expressed by the moments of the reduced density matrix \cite{jfjq},
\begin{equation}
b_{k+1}=\frac{1}{k+1}\sum_{l=0}^{m}(-1)^lb_{k-l}Tr(\rho^{l+1}_A),
\end{equation}
where $b_0=1$ and $k=0,1,...,m-1$. Hence, as the entanglement can be
usually characterized by the reduced density matrix
\cite{wkw,xmys,zwms,zwsm}, it can be also quantified by the moments
of the reduced density matrix. We define the following entanglement
measure based on moments of the reduced states (EMMRS),
\begin{eqnarray*}
E^{rm}(|\psi_{AB}\rangle)=&1-\displaystyle\sum_{i=1}^{\frac{m}{2}}[\frac{4i}{m^2+2m}Tr(\rho^i_A)\\
&+\displaystyle\frac{2m-4i+4}{m^2+2m}Tr(\rho^{i+\frac{m}{2}}_A)]
\end{eqnarray*}
for even $m$, and
\begin{eqnarray*}
E^{rm}(|\psi_{AB}\rangle)&=&1-\sum_{i=1}^{\frac{m-1}{2}}[\frac{4i}{(m+1)^2}Tr(\rho^i_A)
\nonumber\\
&+&\frac{2m-4i+2}{(m+1)^2}Tr(\rho^{\frac{m+1}{2}+i}_A)-\frac{2}{m+1}Tr(\rho^{\frac{m+1}{2}}_A)]
\end{eqnarray*}
for odd $m$.

The EMMRS for general mixed states $\rho_{AB}$ is given by
convex-roof extension,
\begin{equation}
	E^{rm}(\rho_{AB})=\mathop{\min}_{\{p_i,|\psi_i\rangle_{AB}\}}
	\sum_{i}p_iE^{rm}(|\psi_i\rangle_{AB}),
\end{equation}
where the minimization goes over all possible pure state
decompositions of
$\rho_{AB}=\sum_{i}p_i|\psi_i\rangle_{AB}\langle\psi_i|$.

Before presenting our main results, we first prove two lemmas.
\begin{lemma}
For any ensemble $\{p_i,\rho_i\}$ of a quantum state $\rho$, we have
\begin{equation}\label{l1}
Tr[(\sum_{i}p_i\rho_i)^n]\leq \sum_{i}p_iTr(\rho^n_i).
\end{equation}
\end{lemma}

\begin{proof}
We first prove the case for $i=2$, $\rho=p_1\rho_1+p_2\rho_2$, where $p_1+p_2=1$.
For $n\geq 2$, we have
\begin{eqnarray}
Tr(p_1\rho_1+p_2\rho_2)^n &\leq& \{[Tr(p_1\rho_1)^n]^{\frac{1}{n}}+[Tr(p_2\rho_2)^n]^{\frac{1}{n}}\}^n \nonumber\\
                          &=& \{p_1[Tr(\rho_1^n)]^{\frac{1}{n}}+p_2[Tr(\rho_2^n)]^{\frac{1}{n}}\}^n \nonumber\\
                          &\leq& p_1Tr(\rho^n_1)+p_2Tr(\rho^n_2), \nonumber
\end{eqnarray}
where the first inequality is due to the Minkowski inequality, the
second inequality is due to the convexity of the function $y=x^n$ when $x$ is non
negative. By using mathematical induction, we can obtain inequality (\ref{l1}).
\end{proof}

\begin{lemma}
For any ensemble $\{p_i,\rho_i\}$ of a quantum state $\rho$, denote
\begin{equation}
f(\rho)=1-\sum_{i=1}^{\frac{m}{2}}[\frac{4i}{m^2+2m}Tr(\rho^i)
+\frac{2m-4i+4}{m^2+2m}Tr(\rho^{i+\frac{m}{2}})]\nonumber
\end{equation}
for even $m$ and
\begin{eqnarray}
g(\rho) &=& 1- \sum_{i=1}^{\frac{m-1}{2}}[\frac{4i}{(m+1)^2}Tr(\rho^i)\nonumber\\ &+& \frac{2m-4i+2}{(m+1)^2}Tr(\rho^{\frac{m+1}{2}+i})\nonumber\\
         &-& \frac{2}{m+1}Tr(\rho^{\frac{m+1}{2}})]\nonumber
\end{eqnarray}
for odd $m$. We have
\begin{eqnarray}
f(\rho)&\geq&\sum_{i}p_if(\rho_i),\\ \label{15}
g(\rho)&\geq&\sum_{i}p_ig(\rho_i).
\end{eqnarray}
\end{lemma}

\begin{proof}  By definition we have
\begin{small}
\begin{eqnarray}
f(\rho) &=& f(\sum_{j}p_j\rho_j)\nonumber\\
        &=& 1-\sum_{i=1}^{\frac{m}{2}}[\frac{4i}{m^2+2m}Tr[(\sum_{j}p_j\rho_j)^i]\nonumber\\
        &+&\frac{2m-4i+4}{m^2+2m}Tr[(\sum_{j}p_j\rho_j)^{i+\frac{m}{2}}]]\nonumber\\
        &\geq& 1-\sum_{j}p_j\sum_{i=1}^{\frac{m}{2}}[\frac{4i}{m^2+2m}Tr(\rho_j)^i+\frac{2m-4i+4}{m^2+2m}Tr(\rho_j)^{i+\frac{m}{2}}]\nonumber\\
        &=&\sum_{j}p_jf(\rho_j),\nonumber
\end{eqnarray}
\end{small}
where the inequality is due to Lemma 1. Similarly, we can prove the inequality (\ref{15}).
\end{proof}

We are now ready to present a bona fide measure of quantum
entanglement. In fact, a well-defined quantum entanglement measure must satisfy the conditions \cite{vmmp,vvmbp,gvrt} as follows:

(i)\,$E(\rho)\geq 0$ for any quantum state $\rho$ and $E(\rho)=0$ if $\rho$ is separable.

(ii)\,$E$ is invariant under local unitary transformation.

(iii)\,$E$ does not increase on average under stochastic LOCC.

(iv)\,$E$ is convex.

(v)\,$E$ cannot increase under LOCC, that is, $E(\rho)\geq E(\Lambda(\rho))$ for any LOCC map $\Lambda$.

It has been proposed in \cite{wlf} that a covex function $E$ satisfies conditions (v) if and only if it satisfies conditions (ii) and (iii). $E$ is said to be an entanglement monotone\cite{gvid} if the first four conditions hold. From this point of view any entanglement monotone defined in \cite{gvid} could be regarded as a measure  of
entanglement.
\begin{theorem}
For any state $\rho_{AB}$, $E^{rm}(\rho_{AB})$ is a well-defined measure of quantum entanglement.
\end{theorem}
\begin{proof}
Firstly, we prove that if $|\psi_{AB}\rangle$ is a separable pure
state, then $E^{rm}(|\psi_{AB}\rangle)=0$. If $|\psi_{AB}\rangle$ is
a separable state, its reduced density matrix $\rho_A$ is pure. The
moment of any order of $\rho_A$ is equal to $1$, that is,
$Tr(\rho^k_A)=1$, $k=1,2,...$. Thus
$$
E^{rm}(|\psi_{AB}\rangle)=1-\sum_{i=1}^{\frac{m}{2}}\frac{4i}{m^2+2m}+\frac{2m-4i+4}{m^2+2d}=0.
$$ This equation also indicates that when the pure state $|\psi_{AB}\rangle$ is not separable, its reduced state $\rho_A$ is a mixed state, therefore
$Tr(\rho^k_A) < 1$, for $k=1,2,...$. That is
$E^{rm}(|\psi_{AB}\rangle) > 0$. For mixed state $\rho$, by definition and proof of the pure state case, $E^{rm}(\rho_{AB})\geq 0$, and if $\rho_{AB}$ is separable, $E^{rm}(\rho_{AB})=0$.

$E$ is invariant under local unitary transformations from the invariance of $Tr(\rho^i)$.

Below we prove that $E^{rm}(\rho)$ is non-increasing on
average under LOCC. Let $|\psi_{AB}\rangle$ be a bipartite pure state, and $\{\eta_i\}$ be a completely positive trace preserving map on the subsystem $B$. Then the post-mapped state is
\begin{equation}\nonumber
\sigma_i=\frac{1}{p_i}\eta_i(\sigma),
\end{equation}
where $\sigma=|\psi\rangle_{AB}\langle\psi|$ and $p_i=Tr(\eta_i\sigma)$. Let $\sigma^A_i=Tr_B(\sigma_i)$. We have
\begin{equation}\nonumber
\sigma^A=\sum_{i}p_i\sigma^A_i.
\end{equation}
Let $\{p_{ij},\sigma_{ij}\}$ be the optimal ensemble of $\sigma_i$ such that
\begin{equation}\nonumber
E^{rm}(\sigma_i)=\sum_{j}p_{ij}E^{rm}(\sigma_{ij}),
\end{equation}
where $\{\sigma_{ij}\}$ are pure states. Thus,
\begin{eqnarray}
E^{rm}(\rho) &=& f(\sigma^A)\nonumber\\
             &=& f(\sum_{i,j}p_ip_{ij}\sigma^A_{ij})\nonumber\\
             &\geq& \sum_{i,j}p_ip_{ij}f(\sigma^A_{ij})\nonumber\\
             &=& \sum_{i,j}p_ip_{ij}E^{rm}(\sigma_{ij})\nonumber\\
             &=& \sum_{i}p_iE^{rm}(\sigma_i),
\end{eqnarray}
where $\sigma^A_{ij}=Tr_B(\sigma_{ij})$ and the inequality is due to Lemma 2.

Now, for any mixed quantum state $\rho$, let $\{\varepsilon_i\}$ be a completely positive trace preserving map. Then the post-mapped state is
\begin{equation}\nonumber
	\rho_i=\frac{1}{\pi_i}\varepsilon_i(\rho),
\end{equation}
where $\pi_i=Tr(\varepsilon_i\rho)$. Suppose $\{q_j,|\psi_j\rangle\}$ be
the optimal pure-state ensemble of $\rho$. According to the equation
(17), we have
\begin{equation}\label{22}
	E^{rm}(|\psi_j\rangle)\geq \sum_{i}k_{ji}E^{rm}(\rho_{ji}),
\end{equation}
where $k_{ji}=Tr(\varepsilon_i|\psi_j\rangle\langle\psi_j|)$ and $\rho_{ji}=\frac{1}{k_{ji}}\varepsilon_i(|\psi_j\rangle\langle\psi_j|)$. Let $\{k_{jil},|\psi_{jil}\rangle\}$ be the optimal pure-state ensemble of $\rho_{ji}$ such that
$E^{rm}(\rho_{ji})=\sum_{l}k_{jil}E^{rm}(|\psi_{jil}\rangle)$. We have
\begin{eqnarray}
	E^{rm}(\rho) &=& \sum_{j}q_jE^{rm}(|\psi_j\rangle)\nonumber\\
	&\geq& \sum_{j,i}q_jk_{ji}E^{rm}(\rho_{ji})\nonumber\\
	&=& \sum_{j,i,l}q_jk_{ji}k_{jil}E^{rm}(|\psi_{jil}\rangle)\nonumber\\
	&\geq& \sum_{i}\pi_iE^{rm}(\rho_i),\nonumber
\end{eqnarray}
where the first inequality is due to (\ref{22}). The last inequality is due to that
\begin{eqnarray}
	\rho_i &=& \frac{1}{\pi_i}\varepsilon_i(\rho)\nonumber\\
	&=& \frac{1}{\pi_i}\varepsilon_i(\sum_{j}q_j|\psi_j\rangle\langle\psi_j|)\nonumber\\ \label{27}
	&=& \frac{1}{\pi_i}\sum_{j}q_j\varepsilon_i(|\psi_j\rangle\langle\psi_j|)\\
	&=& \frac{1}{\pi_i}\sum_{j}q_jk_{ji}\rho_{ji}\nonumber\\
	&=& \frac{1}{\pi_i}\sum_{j,l}q_jk_{ji}k_{jil}|\psi_{jil}\rangle\langle\psi_{jil}|,\nonumber
\end{eqnarray}
where in the equality (\ref{27}), we have used the linear property of $\varepsilon_i$.

Finally, we prove convexity. Consider $\rho=t\rho_1+(1-t)\rho_2$. Let $\rho_1=\sum_{i}p_i|\psi_i\rangle\langle\psi_i|$ and $\rho_2=\sum_{j}q_j|\phi_j\rangle\langle\phi_j|$ be the optimal pure state decomposition of $E^{rm}(\rho_1)$ and $E^{rm}(\rho_2)$, respectively. Where $\sum_{i}p_i=\sum_{j}q_j=1$ and $p_i$, $q_j>0$. We have
\begin{eqnarray}
E^{rm}(\rho)&\leq&\sum_{i}tp_iE^{rm}(|\psi_i\rangle)+\sum_{j}(1-t)q_jE^{rm}(|\phi_j\rangle)\nonumber\\
            &=&tE^{rm}(\rho_1)+(1-t)E^{rm}(\rho_2),\nonumber
\end{eqnarray}
where the inequality is due to that $\sum_{i}tp_i|\psi_i\rangle\langle\psi_i|+\sum_{j}(1-t)q_j|\phi_j\rangle\langle\phi_j|$ is also a pure state decomposition of $\rho$.
\end{proof}
To demonstrate the usefulness of EMMRS, let us consider the family of $3\times3$ quantum states
 given in Example 1. From our EMMRS we obtain
$$
E^{rm}(\rho_a)=\frac{5}{32}a^2-\frac{5}{32}a+\frac{15}{16}.
$$
The value of $E^{rm}(\rho_a)$ is always greater than 0 for $a\in[\frac{1}{50}(25-\sqrt{141}),\frac{1}{100}(25+\sqrt{141})]$, decreasing with the increase of $a$, see Figure. 3. It is worth noting that in \cite{ga}, the singlet fraction $F^{max}(\rho_a)$, which is directly related to the ability of quantum teleportation, 
also decreases with the increase of $a$. Hence, our entanglement measure also reflects
the ability of the state in quantum teleportation.
\begin{figure}[htbp]
	\centering
	\includegraphics[width=0.5\textwidth]{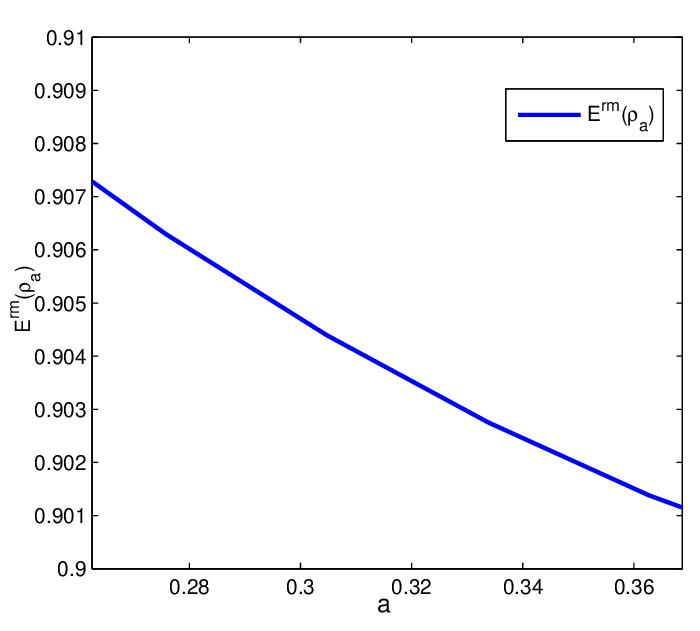}
	\vspace{-2.7em} \caption{$E^{rm}(\rho_a)>0$ for $a\in[\frac{1}{50}(25-\sqrt{141}),\frac{1}{100}(25+\sqrt{141})]$, and $E^{rm}(\rho_a)$ decreases with the increase of $a$.}
	\label{Fig.3}
\end{figure}

From the definition of EMMRS, we see that for $m=2$,
$E^{rm}(|\psi_{AB}\rangle)=\frac{1}{2}(1-Tr(\rho^2_A))=\frac{C^2(|\psi_{AB}\rangle)}{4}$,
which is just the square of concurrence up to a constant factor.
When $m$ increases our entanglement measure can traverse all the
moments of the reduced density matrix $\rho_A$, thus capturing
relatively complete information on the entanglement properties of
quantum states.

Example 5. We consider the following rank-3 states given in \cite{jfjq},
\begin{eqnarray}
|\phi_1\rangle_{AB}&=&\frac{1}{\sqrt{2}}|00\rangle
+\frac{1}{\sqrt{3}}|11\rangle+\frac{1}{\sqrt{6}}|22\rangle,\nonumber\\
|\phi_2\rangle_{AB}&=&\sqrt{\beta_1}|00\rangle+\sqrt{\beta_2}|11\rangle
+\sqrt{1-\beta_1-\beta_2}|22\rangle,\nonumber
\end{eqnarray}
where $\beta_1=\frac{1}{4}$ and $\beta_2=\frac{9+\sqrt{13}}{24}$. The concurrences of these two quantum states are equal, $C(|\phi_1\rangle_{AB})=C(|\phi_2\rangle_{AB})$. However, using our EMRM we obtain $E^{rm}(|\phi_1\rangle_{AB})=0.5139$ and $E^{rm}(|\phi_2\rangle_{AB})=0.5126$. This indicates that although both $|\phi_1\rangle_{AB}$ and $|\phi_2\rangle_{AB}$ are entangled states, the degree of entanglement is different. Our entanglement measure can characterize the entanglement in a more fine way.

\section{VI. Genuine tripartite entanglement measure based on EMMRS}

For a tripartite pure state $|\psi\rangle\in H_A\otimes H_B\otimes
H_C$, we define the genuine tripartite entanglement measure
(GTE-EMMRS) based on EMMRS,
\begin{equation}
E_{GTE}(|\psi\rangle) :=[\prod_{\gamma_i\in \Gamma}E^{rm}(|\psi\rangle_{\gamma_i})]^{\frac{1}{3}},
\end{equation}
where $\Gamma=\{\gamma_i\}$ represents the set of all possible
bipartitions of $\{A, B, C\}$, and the summation goes over all
possible bipartitions $\Gamma=\{(A|B,C), (B|A,C), (C|A,B)\}$.
Generalizing to mixed states $\rho$ via a convex roof extension, we
have
\begin{equation}\label{29}
E_{GTE}(\rho)=\mathop{\min}_{\{p_i,|\psi_i\rangle\}}\sum_{i}p_iE_{GTE}(|\psi_i\rangle),
\end{equation}
where the minimum is obtained over all possible pure state decompositions of $\rho=\sum_{i}p_i|\psi_i\rangle\langle\psi_i|$.

In the following we prove that the GTE-EMMRS is a genuine tripartite
entanglement measure.
\begin{theorem}
The GTE-EMMRS defined in Eq.(\ref{29}) is a genuine tripartite
entanglement measure of tripartite quantum systems.
\end{theorem}

\begin{proof}
The definition of $E_{GTE}(\rho)$ directly implies $E_{GTE}(\rho)=0$ for all biseparable states and $E_{GTE}(\rho)>0$ for all genuine tripartite entangled states.

Next, we prove convexity. For any mixture $\sum_{i}p_i\rho_i$, let $\{p_{ij},\rho_{ij}\}$ be the pure-state ensemble of $\rho_i$. Thus
\begin{eqnarray}
E_{GTE}(\rho)&=&E_{GTE}(\sum_{i}p_i\rho_i)\nonumber\\
             &=&E_{GTE}(\sum_{i,j}p_ip_{ij}\rho_{ij})\nonumber\\
             &\leq&\sum_{i,j}p_ip_{ij}E_{GTE}(\rho_{ij})\nonumber\\
             &=&\sum_{i}p_iE_{GTE}(\rho_i),\nonumber
\end{eqnarray}
where the inequality is due to the definition of $E_{GTE}(\rho)$.

As the EMMRS has been proven to be nonincreasing under LOCC, the
geometric mean of EMMRS for all subsystems is also nonincreasing
under LOCC. Thus $E_{GTE}(\rho)$ is nonincreasing under LOCC.
Therefore, we have completed the proof of the theorem.
\end{proof}

Example 6. Consider the following single parameter family of three-qubit state,
$$
\small{\rho_f=\frac{1}{4f^2+4}\left(\begin{array}{cccccccc}
        1 & \frac{1+f^2}{4} & \frac{f}{4} & 0 &  0 & f & 0 & 1 \\
        \frac{1+f^2}{4} & 1 & 0 & 0 &  0 & 0 & 0 & 0 \\
        \frac{f}{4} & 0 & 2f^2 & 0 & 0 & f^2 & 0 & f \\
        0 & 0 & 0 & 0 & 0 & 0 & 0 & 0 \\
        0 & 0 & 0 & 0 & 0 & 0 & 0 & 0 \\
        f & 0 & f^2 & 0 & 0 & 2f^2 & 0 & \frac{f}{4} \\
        0 & 0 & 0 & 0 & 0 & 0 & 1 & \frac{1+f^2}{4} \\
        1 & 0 & f & 0 & 0 & \frac{f}{4} & \frac{1+f^2}{4} & 1 \\
\end{array}\right),}
$$
where $f\in [0,1]$. By calculation, the GME-concurrence presented in
\cite{mcc} has the form,
	$$C_{GME}(|\psi\rangle):=\mathop{\min}_{\{\gamma_i\in \Gamma\}}\sqrt{2[1-Tr(|\psi\rangle_{\gamma_i}^2)]},$$
which is just a constant, $C_{GME}(\rho_f)=\sqrt{2[1-Tr(\rho_{f\gamma_3})]}=\frac{\sqrt{15}}{4}$ for all $f\in[0,1]$.
However, by using our GTE-EMMRS we obtain
\begin{equation}
E_{GTE}(\rho_f)=\sqrt[3]{\frac{240f^4+465f^2+240}{16384(1+f^2)^2}}.\nonumber
\end{equation}
The genuine tripartite entanglement from our measure depends on the
value of $f$. In other words, our entanglement measure GTE-EMMRS
effectively distinguishes the genuine tripartite entanglement of
this family of quantum states, see Figure. 4. In \cite{xe} the authors proposed an interesting entanglement measure called the concurrence fill, which is given by the area of a triangle composed of three one-to-other bipartite concurrences serving as side lengths:
$$F_{123}=\frac{4}{\sqrt{3}}\sqrt{P(P-C_{1(23)}^2)(P-C_{2(13)}^2)(P-C_{3(12)}^2)},$$
where $P=\frac{1}{2}(C_{1(23)}^2+C_{2(13)}^2+C_{3(12)}^2)$, $C_{i(jk)}$ denotes the concurrence under bipartition $i$ and $jk$. Calculation shows that the concurrence fill decreases with the increase of the parameter $f$. In this sense, GTE-EMMRS and concurrence fill are two inequivalent measures of tripartite entanglement, see Figure. 5.
\begin{figure}[htbp]
\centering
\includegraphics[width=0.45\textwidth]{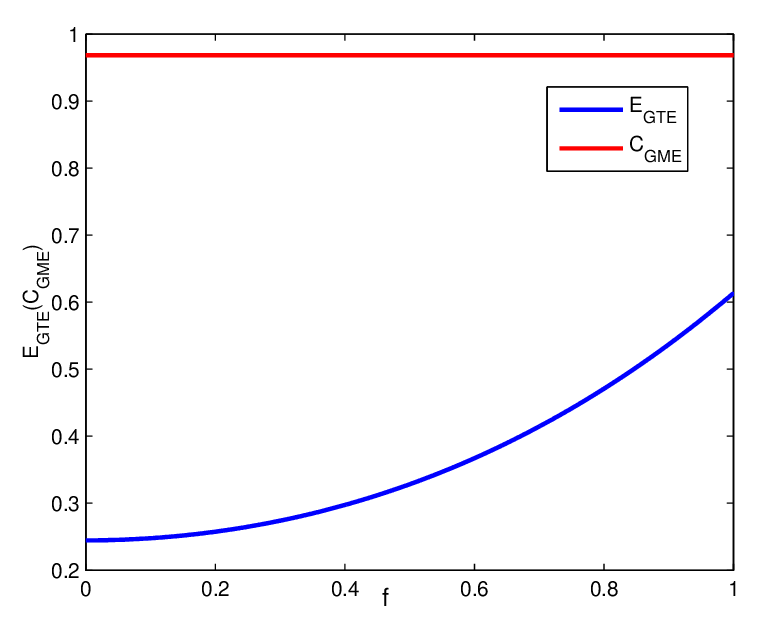}
 \caption{Our entanglement measure $E_{GTE}$ varies
with the $f$ for $f\in [0,1]$, while $C_{GME}$ remains unchanged.}
\label{Fig.4}
\end{figure}
\begin{figure}[htbp]
	\centering
	\includegraphics[width=0.45\textwidth]{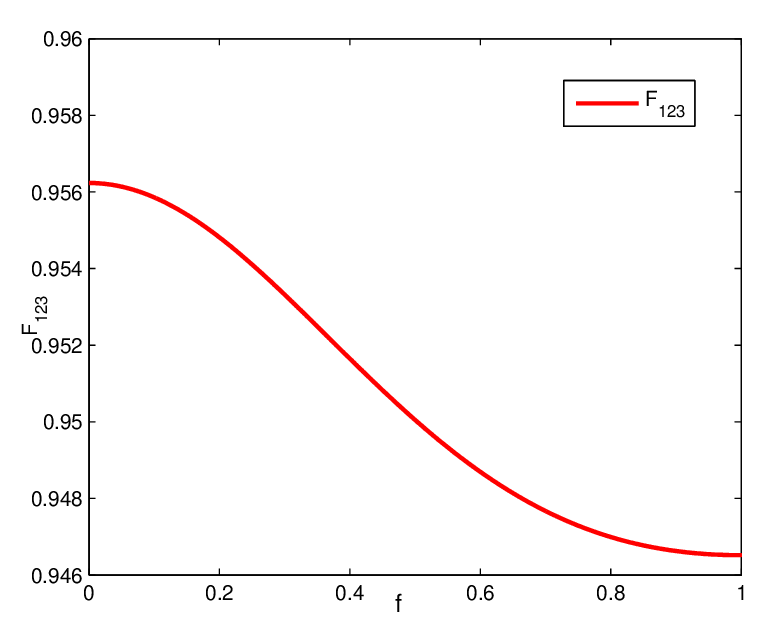}
	\vspace{-1.5em} \caption{Concurrence fill $F_{123}$ versus $f$ for $f\in [0,1]$.}
	\label{Fig.5}
\end{figure}

\section{Conclusions and discussions}

Based on the moments of the realigned matrix of a density matrix we
have proposed an experimentally plausible separability criterion for
any dimensional bipartite quantum states. The main advantage of this
criterion is that it only requires the first two realignment
moments, which simplifies the related experimental measurements. We
have also provided a separability criterion based on the
relationship between the characteristic polynomial coefficients and
the moments of a partially transposed matrix. The discriminant in
this criterion can also be represented in terms of PT moments.
Therefore, this criterion can also be experimentally implemented.
Moreover, we have presented experimentally measurable lower bounds
of concurrence for arbitrary bipartite quantum states, which give
the ways to determine quantitatively the degree of quantum
entanglement without the tomography of unknown quantum states. Based
on the moments of the reduced states, we have also obtained a bona
fide bipartite entanglement measure. Finally, we have presented a
genuine tripartite entanglement measure based on our bipartite
entanglement measure, which discriminates entanglement between
different quantum states that cannot be distinguished by
GME-concurrence.

\bigskip
{\bf Acknowledgments:} ~This work is supported by the Hainan
Provincial Natural Science Foundation of China under Grant
No.121RC539; the National Natural Science Foundation of China under Grant Nos. 12204137, 12075159 and 12171044; the specific research fund of the Innovation Platform for Academicians
of Hainan Province under Grant No. YSPTZX202215; Beijing Natural
Science Foundation (Grant No. Z190005) and the Hainan Provincial Graduate Innovation Research Program (Grant No. Qhys2023-386).

\smallskip
{\bf Data Availability Statement:} This manuscript has no associated
data.

\end{document}